\documentclass[a4paper, USenglish]{lipics-v2018}
\usepackage{amsmath,amsfonts,amssymb}

\usepackage{amscd,amsmath}
\usepackage{amssymb,amsmath,latexsym}

\usepackage{graphicx}

\title{Interval-Like Graphs and Digraphs}

\author{Pavol Hell}{School of Computing Science, Simon Fraser University\\{[Burnaby, B.C., Canada V5A 1S6]}}{pavol@sfu.ca}{}{The author was supported by an NSERC (Canada) Discovery Grant}

\author{Jing Huang}{Mathematics and Statistics, University of Victoria\\{[Victoria, B.C., Canada V8W 2Y2]}} {huangj@uvic.ca}{}{}

 \author{Ross M. McConnell}{Computer Science Department, Colorado State University\\{[Fort Collins, CO 80523-1873}}{rmm@cs.colostate.edu}{}{}

 \author{Arash Rafiey}{Mathematics and Computer Science, Indiana State University\\{[Terre Haute, IN 47809]}}{arash.rafiey@indstate.edu}{}{}

 \keywords{Graph theory}

\subjclass{Graph classes, interval graphs, interval bigraphs}

\authorrunning{P. Hell, J. Huang, R.~M. McConnell, A. Rafiey}

\Copyright{Pavol Hell, Jing Huang, Ross M. McConnell, Arash Rafiey}

\newtheorem{proposition}[theorem]{Proposition}

\EventEditors{Igor Potapov, Paul Spirakis, and James Worrell}
\EventNoEds{3}
\EventLongTitle{43rd International Symposium on Mathematical Foundations of Computer Science (MFCS 2018)}
\EventShortTitle{MFCS 2018}
\EventAcronym{MFCS}
\EventYear{2018}
\EventDate{August 27--31, 2018}
\EventLocation{Liverpool, GB}
\EventLogo{}
\SeriesVolume{117}
\ArticleNo{} 
\nolinenumbers 

\begin{document}

\maketitle

\begin{abstract}
We unify several seemingly different graph and digraph classes under one umbrella.
These classes are all, broadly speaking, different generalizations of interval graphs, and include, in addition 
to interval graphs, adjusted interval digraphs, threshold graphs, complements of threshold tolerance 
graphs (known as `co-TT' graphs), bipartite interval containment graphs, bipartite co-circular arc  graphs, 
and two-directional orthogonal ray graphs. (The last three classes coincide, but have been investigated in 
different contexts.) This common view is made possible by introducing reflexive relationships (loops)
into the analysis.  We also show that all the 
above classes are united by a common ordering characterization, the existence of a min ordering. We 
propose a common generalization of all these graph and digraph classes, namely {\em signed-interval 
digraphs}, and show that they are precisely the digraphs that are characterized by the existence of a min 
ordering. We also offer an alternative geometric characterization of these digraphs. For most of the above 
graph and digraph classes, we show that they are exactly those signed-interval digraphs that satisfy 
a suitable natural restriction on the digraph, like having a loop on every vertex, or having a symmetric edge-set, or being 
bipartite. For instance, co-TT graphs are precisely those signed-interval digraphs that have each edge 
symmetric. We also offer some discussion of future work on recognition algorithms and characterizations.
\end{abstract}

\section{Introduction}

A digraph $H$ is {\em reflexive} if each $vv \in E(H), v \in V(H)$ ($H$ has all loops); {\em irreflexive} if
no $vv \in E(H)$ ($H$ has no loops); and {\em symmetric} if $ab \in E(H)$ implies $ba \in E(H)$.
In this paper, we shall treat both graphs and digraphs; for simplicity we view graphs as symmetric digraphs.
(Thus, graphs can have loops, and irreflexive graphs are loopless.) Loops play an important role in this 
paper, and this is not common in the literature on graph classes that we consider. They allow us to view several seemingly unrelated 
graph classes through a common lens.

A {\em min ordering} of a digraph $H$ is a 
linear ordering $<$ of the vertices of $H$, so that  $ab \in E(H), a'b' \in E(H)$ and $a < a', b' < b$ 
implies that $ab' \in E(H)$.  In other words, a min ordering is an ordering of the vertices such that when the rows and columns of the adjacency matrix are
ordered in this way,
neither the matrix whose rows are $01$ and $11$  nor the matrix whose rows are $01$ and $10$ appears as a submatrix.  (See Figure~\ref{fig:minOrdering}.)
Note that the presence or absence of loops (1's on the diagonal of the adjacency matrix) can affect whether
the graph has a min ordering.

\begin{figure}
	\centerline{\includegraphics[]{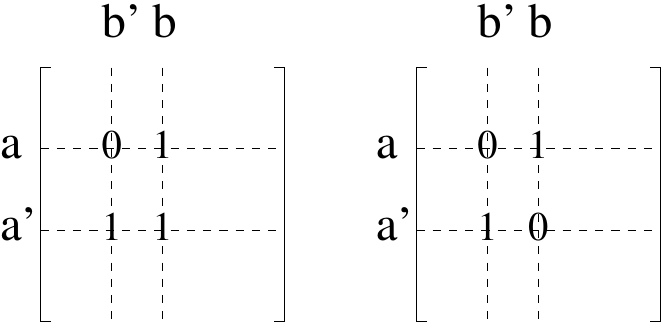}}
	\caption{A min ordering of a digraph is an ordering of the vertices
	such that neither of the depicted submatrices occurs in the corresponding
	adjacency matrix.}\label{fig:minOrdering}
\end{figure}

Our goal in this paper is to promote a class of digraphs (or 0,1-matrices) that is a broad generalization of 
interval graphs and that retains some of the desirable structural properties of interval graphs.  A graph $H$ is an {\em 
interval graph} if it is the intersection graph of a family of intervals on the real line, i.e., if there exists
a family of intervals $\{[x_v,y_v]| v \in V(H)\}$ such that $uv \in E(H)$ if and only if $[x_u,y_u] \cap [x_v,y_v] \neq \emptyset$.
The family of intervals is
an {\em interval model} of $H$. (See Figure~\ref{fig:intvlgraph}.)  We note that the definition implies that an interval graph is 
reflexive.  A related concept for bipartite graphs is as follows. A bipartite graph $H$ with parts $A, B$ is 
an {\em interval bigraph} if there are intervals $\{[x_a,y_a], a \in A\}$, and $\{[x_b,y_b], b \in B\}$, such that 
for $a \in A$ and $b \in B$, $ab \in E(H)$ 
if and only if $[x_a,y_a] \cap [x_b,y_b] \neq \emptyset$.
 
\begin{figure}
	\centerline{\includegraphics[]{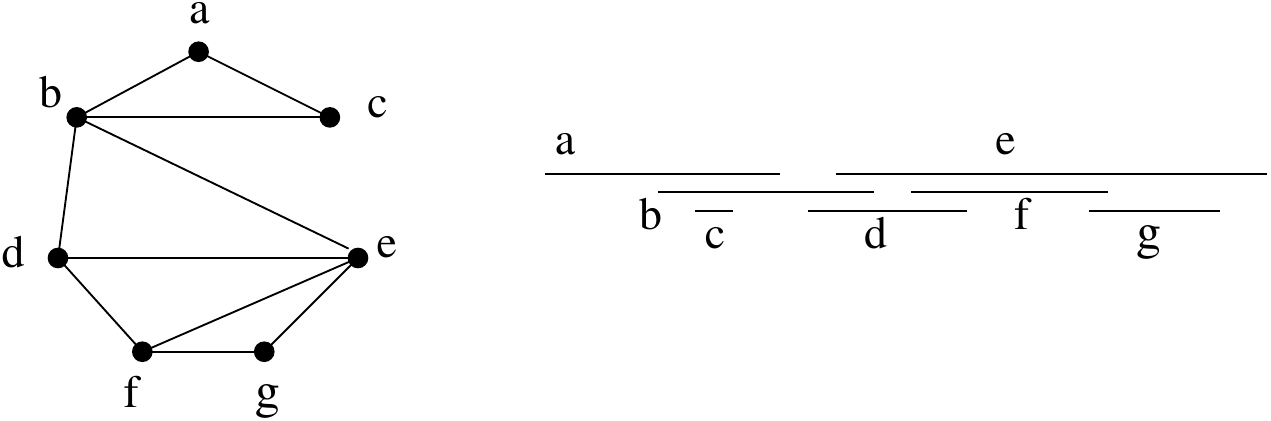}}
	\caption{An interval graph and corresponding interval model.  There is
	an implicit loop at each vertex.}\label{fig:intvlgraph}
\end{figure}

Interval graphs are important in graph theory and in applications, and are distinguished by several elegant
characterizations and efficient recognition algorithms \cite{booth,fh,fg,gol,habib,lekk,spin}. For this reason,
there have been attempts to extend the concept to digraphs \cite{sen-west}, with mixed success. (Many of 
the desirable structural properties are absent.) More recently a more restricted class of digraphs has been found to offer a nicer generalization 
of interval graphs; these are the adjusted interval digraphs \cite{adj}. A digraph $H$ is an {\em adjusted 
interval digraph} if there are two families of real intervals,
the {\em source intervals} $\{[x_v,y_v]| v \in V(H)\}$ and the {\em sink intervals}
and $\{[x_v,z_v]| v \in V(H)\}$ such that $uv \in E(H)$ 
if and only if the source interval for $u$ intersects the sink interval for $v$.
(See Figure~\ref{fig:adjInterval}.)
This differs from the class in~\cite{sen-west} in that the left endpoint, $x_v$, must be shared
by the two intervals $[x_v, y_v]$ and $[x_v,z_v]$ assigned to $v$; they are
``adjusted.''
The interval graphs are the special case where
$[x_v, y_v] = [x_v,z_v]$ for each $v \in V(H)$.
An {\em adjusted interval model} of $H$ is a set of source and sink intervals
that represent $H$ in this way.

\begin{figure}
	\centerline{\includegraphics[]{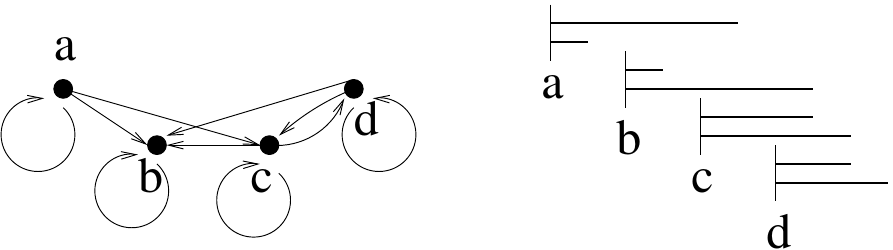}}
	\caption{An adjusted interval digraph and a corresponding adjusted interval model.
	The source interval for each vertex is the upper one.}\label{fig:adjInterval}
\end{figure}

Adjacency on a set of intervals can also be defined by interval containment.  A 
graph is a {\em containment graph of intervals} \cite{spin} if there is a family of intervals 
$\{[x_v,y_v]|v \in V(H)\}$ on the real line such that $uv \in E(H)$
if and only if one of $[x_u,y_u]$ and $[x_v,y_v]$ contains the other. A graph is a containment
graph of intervals if and only if it and its complement are both {\em transitively orientable}, thus if and only if it is a 
{\em permutation graph} \cite{spin}. 

For this paper, a more relevant class is a bipartite version of this concept. 
A bipartite graph 
$H$ with parts $A, B$ is an {\em interval containment bigraph} \cite{spin} if there are sets of intervals 
$\{I_a| a \in A\}$, and 
$\{J_b| b \in B$\}, such that $ab \in E(H)$ if and only if $J_b \subseteq I_a$. These graphs have been studied,
from the point of view of another geometric representation, as two-directional orthogonal ray graphs \cite{ueno}. 
A bipartite graph $H$ with parts $A$ and $B$ is called a {\em two-directional orthogonal ray graph} if there exists 
a set $\{U_a, a \in A\}$ of upwards vertical rays, and a set $\{R_b, b \in B\}$ of horizontal rays to the right such that 
$ab \in E(H)$ if and only if $U_a \cap R_b \neq \emptyset$. It is known that a bipartite graph is an interval 
containment graph if and only if it is a two-directional orthogonal ray graph \cite{jingx}, and if and only if its
complement is a circular arc graph \cite{fhh}. 

It is sometimes convenient to view bipartite graphs as digraphs, with all edges oriented from part $A$ to 
part $B$; thus we speak of a {\em bipartite interval containment digraph}, a {\em bipartite interval digraph}, 
or a {\em two-directional orthogonal ray digraph}. In general, a {\em bipartite digraph} is a bipartite graph 
with parts $A$ and $B$ and all arcs being oriented from $A$ to $B$.

There is an interesting intermediate concept that uses both intersection and containment of intervals to define 
adjacency.
An interval model of an interval graph $G$ can be viewed as
two mappings $\{v \to x_v| v \in V(H)\}$ and $\{v \to y_v| v \in V(H)\}$ such that $x_v \leq y_v$
for each $v \in V(H)$, and such that $uv \in E(H)$ if and only if $y_v \leq x_u$ and $y_u \leq x_v$.  
The constraint $x_v \leq y_v$ comes from the need for $[x_v,y_v]$ to be an interval.
The proposition that two intervals intersect is the same as $x_v \leq y_u$ and $x_u \leq y_v$, since this
means that neither interval lies entirely to the right of the other.

A generalization of interval models is obtained by dropping the constraint $x_v \leq y_v$ in this formulation.
To develop the motivation for this, we start with the complements of threshold tolerance graphs. A graph $H$ is a {\em threshold tolerance 
graph} \cite{bruce} if its vertices $v$ can be assigned weights $w_v$ and tolerances $t_v$ so that $ab$ is an 
edge of $H$ if and only if $w_a+w_b > t_a$ or $w_a+w_b > t_b$. (When all $t_v$ are equal, this defines a better 
known class of {\em threshold graphs} \cite{hammer}.) {\em Co-threshold tolerance (`co-TT') graphs} are complements of
threshold tolerance graphs. Equivalently, a graph $H$ is a co-TT graph, if 
there exist real numbers $x_v, y_v, v \in V(H)$, such that $ab \in E(H)$ if and only if $x_a \leq y_b$ and 
$x_b \leq y_a$~\cite{split}.  This differs from the definition of interval graphs in that it is no longer
required that $x_v \leq y_v$, illustrating the motivation for dropping the constraint
in this case.  (See Figure~\ref{fig:cottExample}.)
That these are precisely the co-TT graphs is easily seen by letting $x_v=w_v$ and $y_v=t_v-w_v$. The two mappings $v \to x_v$ 
and $v \to y_v$, are called the {\em co-TT model} of $H$. 

\begin{figure}
	\centerline{\includegraphics[]{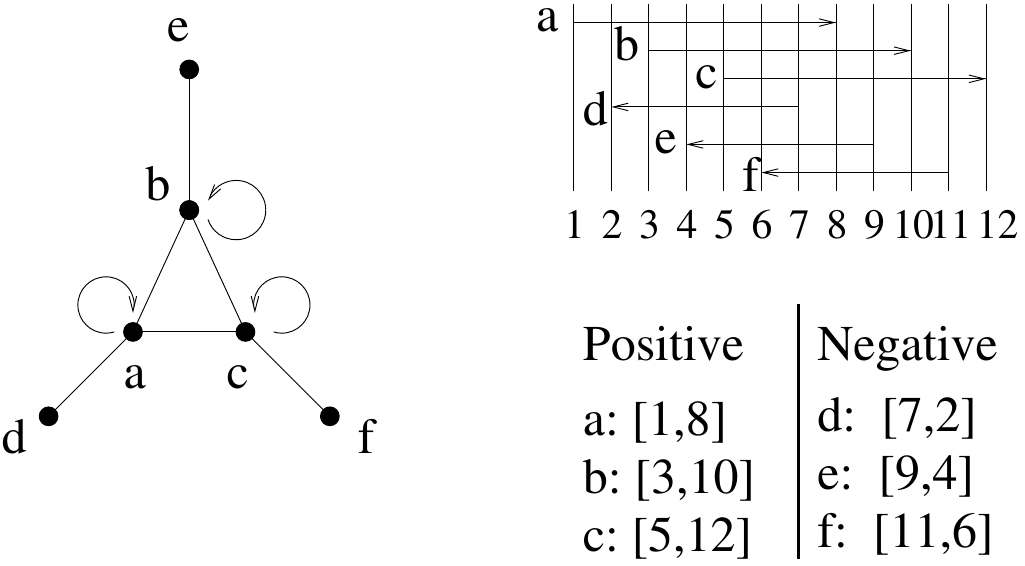}}
	\caption{A co-TT graph and a corresponding co-TT model; $ab$ is an
	edge since $1 \leq 10$ and $3 \leq 8$, $ad$ is an edge since
	$1 \leq 2$ and $7 \leq 8$.  However, $bd$ is not an edge: although
	$7 \leq 10$, $3$ is not less than or equal to $2$.  
	The example of this figure is one of the well-known minimal
	graphs that are not interval graphs, illustrating that
	the interval graphs are a proper subclass of the co-TT graphs.}\label{fig:cottExample}
\end{figure}

One view of a co-TT model is that
there are now intervals whose `beginning' $x_v$ may
come after their `end' $y_v$. In other words, we may have `intervals' $[x_v,y_v]$ with $y_v < x_v$. We may 
view a co-TT model as consisting of intervals $[x_v,y_v], v \in V(H),$ some of which go in the positive direction
(have $x_v \leq y_v$) and others go in the negative direction (have $x_v > y_v$). We speak of {\em positive}
or {\em negative} intervals, and {\em positive} or {\em negative} vertices that correspond to them. 
(In the literature \cite{golovach,split,ross,jing,bruce},
the direction is denoted by {\em colors} of the intervals: positive intervals, and vertices, are 
colored {\em blue}, and negative intervals, and vertices, are colored {\em red}.) The above definition of adjacency 
has interesting consequences. Two positive vertices are adjacent if and only if they intersect; in particular, each positive 
 vertex has a loop. Two negative vertices are never adjacent; in particular negative vertices have no loops. Finally, a positive
 vertex $u$ corresponding to a positive interval $[a,b]$ and a negative vertex $v$ corresponding to a negative interval $[c,d]$
are adjacent if and only if $[d,c]$ is contained in $[a,b]$ (i.e., $a \leq d \leq c \leq b$). We also use the following
signed shorthand, which will be useful later: a positive vertex or interval will be called a {\em $+$-vertex} or
{\em $+$-interval} respectively, and a negative vertex or interval will be called a {\em $-$-vertex} or
{\em $-$-interval} respectively. It follows from the above discussion that in a co-TT graph, the $+$-vertices
induce a reflexive interval graph, the $-$-vertices form an independent set, and the edges between the 
$+$-vertices and the $-$-vertices form a bipartite interval containment graph.

Note that co-TT graphs are a generalization of interval graphs; the interval graphs are those co-TT graphs where all vertices are positive.
In other words, they are the reflexive co-TT graphs.

\section{Signed Interval Digraphs}

We have now seen extensions of interval graphs in two different directions. First, by taking two (adjusted) intervals 
instead of just one interval, we were able to extend the definition from reflexive graphs to reflexive digraphs.
Second, by admitting intervals $[a,b]$ that go in the negative direction (have $b < a$), we were able to extend 
the definition from reflexive graphs to graphs that have some vertices with loops and others without. 
Both these generalizations have proved very fruitful \cite{adj,fh,adj,golovach,arxiv,split,ross,jing,bruce}. 

We now define a new class of digraphs that unifies these extensions.
A digraph $H$ is a {\em signed-interval digraph} if there exist three mappings from
$V(H)$ to the real line,
$v \to x_v, v \to y_v$, and $v \to z_v$, such that $uv \in E(H)$ if and only if $x_u \leq z_v$ and $x_v \leq y_u$.  
We call the three mappings $v \to x_v, v \to y_v$, and $x \to z_v$ a {\em signed-interval model} of $H$.
Alternatively, a signed interval model is obtained in by assigning,
for each $v \in V(H)$ a {\em source interval}
$[x_v,y_v]$ and a {\em sink interval} $[x_v,z_v]$, such that
$uv \in E(H)$ if and only if $x_u \leq z_v$ and $x_v \leq y_u$.  (See figure~\ref{fig:signedInterval}.)
Since it is possible that $x_v > y_v$ and/or $x_v > z_v$, each
of $[x_v,y_v]$ and $[x_v,z_v]$ can be negative or positive.  Since the source
interval and sink interval for $v$ share the endpoint $x_v$, we retain
the property that the intervals are adjusted.

Signed-interval digraphs with all intervals positive, are reflexive, and are adjusted 
interval digraphs. Signed-interval digraphs with $y_v = z_v$, for all $v \in V(H)$, are symmetric, and are
co-TT graphs. Signed-interval digraphs that satisfy both conditions, i.e., with all 
$x_v \leq y_v = z_v, v \in V(H)$, are interval graphs. Furthermore, we show below that there
are no reflexive signed-interval digraphs other than adjusted interval digraphs, no symmetric
signed-interval digraphs other than co-TT graphs, and no reflexive and symmetric signed-interval 
digraphs other than interval graphs.

\begin{figure}
	\centerline{\includegraphics[]{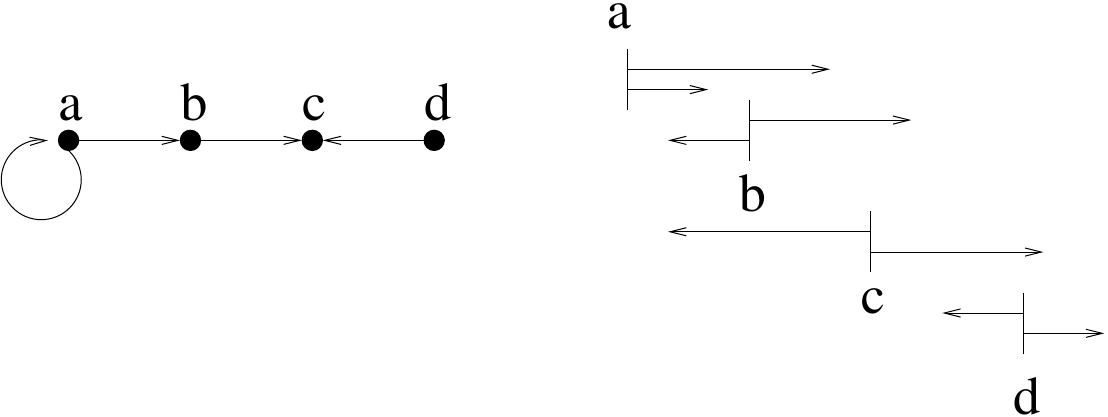}}
	\caption{A signed interval digraph and a corresponding signed interval model.  The source interval for each vertex is the upper one.  There is a loop
	at $a$ because its positive source interval intersects its positive
	sink interval.  There is an edge from
	$a$ to $b$ because $a$'s positive source interval contains $b$'s
	negative sink interval, an edge from $b$ to $c$ because
	$b$'s positive source interval intersects $c$'s positive sink
	interval, and an edge from $d$ to $c$ because $d$'s negative
	source interval is contained in $c$'s positive sink interval.}~\label{fig:signedInterval}
\end{figure}

The structure of signed-interval digraphs can be described in a language similar to what was used
for co-TT graphs. Let $H$ be a signed-interval digraph and consider a signed-interval 
model of $H$ given by the ordered pairs $(I_v,J_v)$ of intervals where $I_v = [x_v,y_v]$ 
and $J_v = [x_v,z_v]$. For $\alpha, \beta \in \{+,-\}$, we say a vertex $v$ is of type $(\alpha,\beta)$ if 
$I_v$ is an $\alpha$-interval and $J_v$ is a $\beta$-interval. The subdigraph of $H$ induced by 
$(+,+)$-vertices is an adjusted interval digraph. The $(-,-)$-vertices of $H$ form an independent set. 
The arcs between the  $(+,-)$- and $(-,-)$-vertices form a bipartite interval containment digraph. 
The arcs between the $(-,+)$- and $(-,-)$-vertices also form a bipartite interval containment digraph.
Similar properties hold for the other parts and their connections. 

We emphasize that our definition of co-TT graphs differs from the standard definition \cite{golovach,split,bruce}. 
In the standard definition, the condition $ab \in E(H) \iff  x_a \leq y_b$ and $x_b \leq y_a$ is applied only for 
$a \neq b$, and so the graphs have no loops. Thus a graph under the standard interpretation is co-TT if and only 
if with a suitable addition of loops it is co-TT under our definition above. This difference is not important as it 
was shown in \cite{gol} that if a graph $H$ is co-TT (in the standard sense), then it has a co-TT model with 
negative intervals for all simplicial vertices without true twins and all other intervals positive. Thus there is an easy 
translation between the co-TT graphs as defined here and the standard irreflexive co-TT graphs:
namely, loops are to be placed on all vertices other than simplicial vertices without true twins.


\section{Min Orderings}

Interval graphs, adjusted interval digraphs, co-TT graphs, and two-directional 
orthogonal ray digraphs all have min orderings when care is taken to specify which vertices have loops 
and which do not.~\cite{fh,adj,esa2012,ueno}.  

Min orderings are a useful tool for graph homomorphism problems.
A {\em homomorphism} of a digraph $G$ to a digraph $H$ is a mapping $f: V(G) \to V(H)$ such that 
$f(u)f(v) \in E(H)$ whenever $uv \in E(G)$. If a digraph $H$ has a min ordering, there is a simple
polynomial-time algorithm to decide if a given input graph $G$ admits a homomorphism to a fixed
digraph $H$ \cite{emmo,hombook}. In fact, the algorithm is well known in the AI community as the
arc-consistency algorithm \cite{hombook}; it is easy to see that it also solves {\em list homomorphism}
problems, where we seek a homomorphism of input $G$ to fixed $H$ taking each vertex of $G$ to
one of a `list' of allowed images. In fact, many (but not all) homomorphism and list homomorphism
problems that can be solved in polynomial time can be solved using arc-consistency with respect to 
a min ordering.

Graph and digraph homomorphism problems are special cases of constraint 
satisfaction problems. A general tool for solving polynomial time solvable 
constraint satisfaction problems are the so-called polymorphisms~\cite{buljeav}. 
Without going into the technical details, we mention that min-orderings 
are equivalent to conservative semilattice polymorphisms~\cite{adj}.

We prove below that a digraph has a min ordering if and only if it is a signed-interval digraph.
We also give another geometric characterization of signed-interval digraphs, as bi-arc
digraphs. We show that a {\em reflexive} signed-interval digraphs are precisely adjusted 
interval digraphs, that {\em symmetric} signed-interval digraphs are precisely co-TT graphs, 
that {\em reflexive and symmetric} signed-interval digraphs are precisely interval graphs, 
and that {\em bipartite} signed-interval digraphs are precisely two-directional ray graphs.

The main result of this section is the following.

\begin{theorem}\label{main}
A digraph admits a min ordering if and only if it is a signed-interval digraph.
\end{theorem}

Before embarking on the proof we offer an alternate definition of a min ordering. Consider any linear ordering $<$ 
of $V(H)$.  To this ordering, we prepend an intial element $\alpha$,
which is a place holder and not a vertex.  Thus, $\alpha < x$ for
each vertex $x$.  We denote by $O(a)$ the last 
vertex $b$ (in the order $<$), such that $b$ is an out-neighbor of $a$ 
(i.e., such that $ab \in E(H)$), or $\alpha$ if $a$ has no out-neighbor.
Similarly, for each vertex $b$, 
we denote by $I(b)$ the last vertex $a$ such 
that $a$ is an in-neighbor of $b$ (i.e., such that $ab \in E(H)$),
or $\alpha$ if $a$ has no in-neighbor.

\begin{proposition}\label{nolem}
A linear ordering $<$ of $V(H)$ is a min ordering of a digraph $H$ if and only if the following property holds:

\begin{center}
 $ab \in E(H)$ if and only if $a \leq I(b)$ and $b \leq O(a)$.
\end{center}
\end{proposition}

\begin{proof}
Suppose first that $<$ is a min ordering of $H$ with $\alpha$ prepended.
If $ab \in E(H)$, then by the definition of $O(a), I(b)$ we have $a \leq I(b)$ and $b \leq O(a)$.
On the other hand, let $a \leq I(b)$ and $b \leq O(a)$. Note that if $a=I(b)$ or $b=O(a)$ we have $ab \in E(H)$ also 
by definition. Therefore it remains to consider vertices $a, b$ such that $a < c=I(b)$ and $b < d=O(a)$. Then
$ad, cb \in E(H)$ and the min ordering property implies that $ab \in E(H)$. This proves the property.

Conversely, assume that $<$ is a linear ordering of $V(H)$ with $\alpha$ prepended
and that the property holds for $<$. We claim 
it is a min ordering of $H$. Otherwise some $ab \in E(H), a'b' \in E(H), a < a', b' < b$ would have 
$ab' \not\in E(H)$. This is a contradiction, since we have $a < a' \leq I(b')$ and $b' < b \leq O(a)$.
\end{proof}

We proceed to prove the theorem.

\begin{proof}
Suppose $<$ is a min ordering of a digraph $H$ with
$\alpha$ prepended.  We represent each vertex $v \in V(H)$ by the mappings 
$v \to v, v \to O(v), v \to I(v)$. In other words, $v$ is represented by the two intervals $[v,O(v)]$ and $[v,I(v)]$. 
It follows from Proposition \ref{nolem} that $ab \in E(H)$ if and only if $a \leq I(b)$ and $b \leq O(a)$. 
Thus 
$H$ is a signed-interval digraph.

Conversely, suppose we have the three mappings $v \to x_v, v \to y_v, v \to z_v$ from $V(H)$ to the real line, such that $ab \in E(H)$ if and 
only if $x_a \leq z_b$ and $x_b \leq y_a$. Without loss of generality we may assume the points $\{x_v| v \in V(H)\}$ 
are all distinct. Then we claim that the left to right ordering of the points $x_v$ yields a min ordering $<$ of $H$. 
(Specifically, we define $a < b$ if and only if $x_a$ precedes $x_b$.) Consider now $ab \in E(H), a'b' \in E(H),$
with $a < a', b' < b$. This means that $x_a < x_{a'} \leq z_{b'}$ and $x_{b'} < x_b \leq y_a$, whence we must
have $ab' \in E(H)$. 
\end{proof}

\section{An alternate geometric representation}

Digraphs that admit a min ordering have another 
geometric representation. Let $C$ be a circle with two distinguished points (the {\em poles}) $N$ and $S$, 
and let $H$ be a digraph. Let $I_v, v \in V(H)$ and $J_v, v \in V(H)$ be two families of arcs on $C$ such 
that each $I_v$ contains $N$ but not $S$, 
and each $J_v$ contains $S$ but not $N$. We say that the families 
$I_v$ and $J_v$ are {\em consistent} if they have the same clockwise order of their clockwise ends, i.e., the 
clockwise end of $I_a$ precedes in the clockwise order the clockwise end of $I_b$ if and only if the clockwise 
end of $J_a$ precedes in the clockwise order the clockwise end of $J_b$. Suppose two families $I_v, J_v$ 
are consistent; we define an ordering $<$ on $V(H)$ where $a < b$ if and only if the clockwise end of $I_a$ 
precedes in the clockwise order the clockwise end of $I_b$; we call $<$ the ordering {\em generated by} the
 consistent families $I_v, J_v$.

A {\em bi-arc model} of a digraph $H$ is a consistent pair of families of circular arcs, 
$I_v, J_v, v \in V(H)$, such that $ab \in E(H)$ if and only if $I_a$ and $J_b$ are disjoint. 
A digraph $H$ is called a {\em bi-arc digraph} if it has a bi-arc model.

\begin{theorem}\label{bia}
A digraph $H$ admits a min ordering if and only if it is a bi-arc digraph. 
\end{theorem}

\begin{proof}
Suppose $I_v, J_v$ form a bi-arc model of $H$. We claim that the ordering $<$ generated 
by $I_v, J_v$ is a min ordering of $H$. Indeed, suppose $a < a'$ and $b' < b$ have $ab, a'b' \in E(H)$. 
Then $I_{a'}$ spans the area of the circle between $N$ and the clockwise end of $I_a$, and $J_b$ spans 
the area of the circle between $S$ and the clockwise end of $J_{b'}$. (See Figure 1.) This implies that 
$I_a$ and $J_{b'}$ are disjoint: indeed, the counterclockwise end of $I_a$ is blocked from reaching $J_{b'}$ 
by $J_b$ (since $ab \in E(H)$), and the counterclockwise end of $J_{b'}$ is blocked from reaching $I_a$ by 
$I_{a'}$ (since $a'b' \in E(H)$). (The clockwise ends are fixed by the ordering $<$.)

\begin{figure}
\begin{center}
\includegraphics[height=4cm]{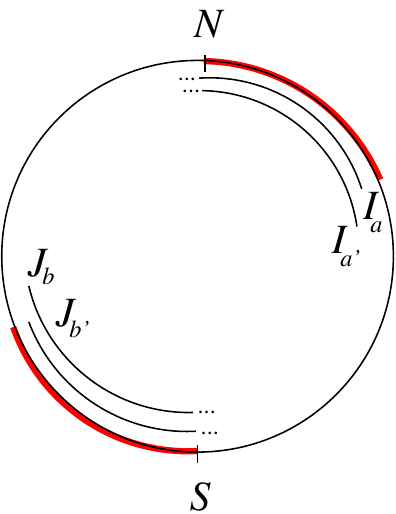}
\caption{Illustration for the proof of Theorem \ref{bia}}
\end{center}
\end{figure}

Conversely, suppose $<$ is a min ordering of $H$. We construct families of arcs $I_v$ and $J_v$, with 
$v \in V(H)$, as follows. The intervals $I_v$ will contain $N$ but not $S$, the intervals $J_v$ will contain 
$S$ but not $N$. The clockwise ends of $I_v$ are arranged in clockwise order according to $<$, as are the 
clockwise ends of $J_v$. The counterclockwise ends will now be organized so that $I_v, J_v, v \in V(H)$,
becomes a bi-arc model of $H$. For each vertex $v \in V(H)$, we define
$O(v)$ and $I(v)$ as in the proof of Theorem 1.
Then we assign the counterclockwise endpoint of $I_v$ to be N if $v$
has no out-neighbors, or else extend 
$I_v$ counterclockwise as far as possible without intersecting $J_{O(v)}$, and 
assign the the counterclockwise endpoint of each $J_v$ to be $S$
if $v$ has no in-neighbors, or else extend $J_v$ 
counterclockwise as far as possible without intersecting $I_{I(v)}$. We claim this is a bi-arc model 
of $H$. Clearly, if $b > O(a)$, then $I_a$ intersects $J_b$ by the construction, and similarly for $a > I(b)$ 
we have $J_b$ intersecting $I_a$. This leaves disjoint all pairs $I_a, J_b$ such that $a \leq I(b)$ and 
$b \leq O(a)$; since $aO(a), I(b)b \in E(H)$, the definition of min ordering implies that $ab \in E(H)$, as 
required.  
\end{proof}

\begin{corollary}\label{rich}
The following statements are equivalent for a digraph $H$.
\begin{itemize}
\item
$H$ has a min ordering
\item
$H$ is a signed-interval digraph
\item
$H$ is a bi-arc digraph.
\end{itemize}
\end{corollary}

\section{0,1-Matrices and bipartite graphs}\label{matx}

Irreflexive graphs with at least one edge do not admit a min ordering, since the vertices that are not
reflexive form an independent set.  However, in the special
case of bipartite graphs, a version of min ordering has been studied, and has yielded interesting examples. We
will describe that version below, but we first want to explain how to view that definition as a special case of
min ordering as defined here.

A useful perspective on min orderings is obtained by considering 0,1-matrices. 
Square 0,1-matrices naturally correspond to adjacency matrices of digraphs. 
Let a {\em simultaneous permutation} of rows and columns of a matrix be one where
the permutation of the rows is the same as the permutation of the columns.
An {\em independent permutation} of rows and columns allows the permutation
of the rows to be different from the permutation of the columns.  

Let $L$ be the two by two matrix with rows 
$0 1$ and $1 1$, and let $K$ be the two by two matrix with rows $0 1$ and $1 0$. (These have been 
given other names in the literature, up to a simultaneous permutation of rows and columns they are 
the gamma matrix, and the identity matrix.) A matrix $M$ is called {\em $K, L$-free} if it does not contain 
$K$ or $L$ as a submatrix. If $M$ is the adjacency matrix of a digraph $H$, and if the rows and columns 
of $H$ are in the order $<$, then $M$ is $K, L$-free if and only if $<$ is a min ordering. 
We call a $M$ a {\em min-orderable matrix} if its rows and columns 
can be simultaneously permuted to produce a $K, L$-free matrix.  A digraph has a min ordering if and only 
if its adjacency matrix is min-orderable.

Therefore, we can say much about matrices that are min-orderable.

\begin{theorem}\label{veta}
A square 0,1-matrix is min-orderable if and only if it is the adjacency matrix of a signed-interval 
digraph.
\end{theorem}

Another natural interpretation of a 0,1-matrix is that it represents adjacencies in a bipartite graph, with 
rows corresponding to one part and columns to the other part. 
The {\em bi-adjacency matrix} of a bipartite
graph $H$ with marts $A, B$ has its $i,j$-th entry is $1$ if and only if the $i$-th vertex in $A$ is 
adjacent to the $j$-th vertex in $B$. Note that for this interpretation it is not required that the 
matrix be square. For matrices that are not necessarily square, we can still ask for independent
permutations of rows and columns, to produce a $K, L$-free matrix. This suggests a definition of 
min ordering for bipartite graphs as follows. A {\em min ordering of a bipartite graph} $H$ with 
parts $A$ and $B$ is a linear ordering $<_A$ of $A$ and a linear ordering $<_B$ of $B$ so that for any 
$a, a' \in A, b, b' \in B$ such that $ab \in E(H), a'b' \in E(H)$ and $a < a', b' < b$ we have $ab' \in E(H)$.
This is the definition that has been used in the literature; it is clear how it avoids the problems of the
general definition.

There is a simple transformation that connects the two interpretations of 0,1-matrices. For a matrix $M$ 
with $k$ rows and $\ell$ columns, we define the $(k+\ell)$ by $(k+\ell)$ square matrix $M^+$ to contain 
the matrix $M$ in the first $k$ rows and the last $\ell$ columns, with $0$ everywhere else. Then a 
simultaneous row/column permutation of $M^+$ corresponds to independent row and column 
permutations of $M$. Note that the square matrix $M^+$ is an adjacency matrix of the digraph obtained
from $H$ by directing all edges from $H$ from the first part to the second part. Thus to view the special
definition of a min ordering for bipartite graphs as a particular case of the general definition, it suffices to
view bipartite graphs as digraphs with all edges oriented from the first part to the second part. We shall 
say that $H$ is a {\em bipartite digraph} if it is obtained from a bipartite graph in this way.

A robust class of bipartite graphs is relevant for our discussion. A bipartite graph $H$ with parts $A$
and $B$ is called a {\em two-directional orthogonal ray graph} if there exists a set $U_a, a \in A,$ of upwards 
vertical rays, and a set $R_b, b \in B,$ of horizontal rays to the right such that $ab \in E(H)$ if and
only if $U_a \cap R_b \neq \emptyset$. Note that we may, if needed, view a two-directional orthogonal ray graph
as a bipartite digraph, with all edges oriented from (say) vertical rays to horizontal rays.


The following theorem is obtained by a combination of results from \cite{fhh,jingx,ueno}.

\begin{theorem}\label{bipo}
The following statements are equivalent for a bipartite graph $H$.
\begin{itemize}
\item
$H$ is a two-directional orthogonal ray graph
\item
the complement of $H$ is a circular arc graph
\item
$H$ is an interval containment graph.
\end{itemize}
\end{theorem}

Matrices that can be permuted to avoid small submatrices have been of much interest \cite{farber,woeg,anna}.
This of course corresponds to characterizations of digraphs by forbidden ordered subgraphs \cite{damas,bojan}.
Our focus was on $K, L$-free matrices. Let the matrix $\Gamma$ be obtained from $L$ by simultaneously 
exchanging the rows and columns; i.e., $\Gamma$ has rows $1 1$, $1 0$. Let $I$ be the two by two identity
matrix. It is easy to see that considering $I, \Gamma$-free matrices is equivalent to considering $K, L$-free 
matrices, as the permutation that simultaneously reverses rows and columns of matrix $M$ transforms a
$I, \Gamma$-free matrix to a $K, L$-free matrix and vice versa. Matrices that are $\Gamma$-free have been
intensively studied \cite{farber,anna}, cf. \cite{spin}. A bipartite graph $H$ is {\em chordal bipartite} if it contains 
no induced cycle other than $C_4$. A reflexive graph is {\em strongly chordal} if it contains no induced cycle or 
induced trampoline. (A {\em trampoline} is a complete graph on $v_0, v_1, v_2, \dots, v_{k-1}, k > 2$ with vertices 
$u_i, i = 0, 1, \dots, k-1,$ each only adjacent to $v_i, v_{i+1}$, subscripts computed modulo $k$.) The adjacency
matrix of a reflexive graph $H$ can be made $\Gamma$-free by simultaneous row / column permutations if and only 
if $H$ is strongly chordal; the bi-adjacency matrix of a bipartite graph $H$ can be made $\Gamma$-free by independent
permutations of rows and columns if and only if $H$ is chordal bipartite \cite{farber}. These results amount to
forbidden structure characterizations of matrices that are permutable (by simultaneous or independent row and
column permutations) to a $\Gamma$-free format. Algorithms to recognize such matrices efficiently have been
given in \cite{anna,tarjan}. For $L$-free matrices, or equivalently, for $I$-free matrices a forbidden structure
characterization is given in \cite{jing}. An $O(n^2)$ recognition algorithm is claimed in \cite{rus}, cf. \cite{spin}.

\section{Special cases}

We now explore what min orderings look like in the special cases we have discussed, namely reflexive
graphs, reflexive digraphs, undirected graphs, and bipartite graphs. The results are all corollaries 
of Theorem \ref{main} and Proposition \ref{nolem}. 

\begin{corollary}
A reflexive graph $H$ is a signed-interval digraph if and only if it is an interval graph.
\end{corollary}

\begin{corollary}
A reflexive digraph $H$ is a signed-interval digraph if and only if it is an adjusted interval digraph.
\end{corollary}

Next we focus on symmetric digraphs, i.e., graphs.

\begin{corollary}\label{lookhere}
A graph $H$ is a signed-interval digraph, i.e., has a min ordering, if and only if it is a co-TT graph.
\end{corollary}

\begin{proof}
Consider a co-TT model of $H$, given by the mappings $v \to x_v, v \to y_v$, setting the third mapping $v \to z_v$
with each $z_v = y_v$, yields a signed-interval digraph model of $H$. Conversely, assume $H$ is a 
graph, i.e., a symmetric digraph, that is a signed-interval digraph. Let $<$ be a min ordering of $H$; 
we again have $O(v)=I(v)$ for all vertices $v$. We claim that the mappings $v \to x_v=v, v \to y_v=O(v)$ define 
a co-TT model. Indeed, from Proposition \ref{nolem} we have $ab \in E(H)$ if and only if $a \leq O(b)=y_b$ and 
$b \leq O(a)=y_a$, as required.
\end{proof}

Finally, for bipartite graphs we have the following result, stated for convenience in the language of bipartite digraphs.
Note that this has only one consequence, that is, when we consider an edge $ab \in E(H)$ we always assume $a \in A$
and $b \in B$. 

\begin{corollary}
A bipartite digraph $H$ is a signed-interval digraph, i.e., has a min ordering, if and only if it is a two-directional orthogonal ray graph.
\end{corollary}

Note that two-directional orthogonal ray graphs themselves have two other equivalent characterizations in Theorem \ref{bipo}.
The characterization of two-directional orthogonal ray graph by the existence of a min ordering is also observed in \cite{esa2012,ueno}.

\begin{proof}
Suppose $H$ has a signed-interval model given by the three mappings 
$v \to x_v, v \to y_v, v \to z_v$ such that 
$ab \in E(H)$ if and only if $x_a \leq z_b$ and $x_b \leq y_a$. We construct a two-directional ray model for $H$ as 
follows. For each $a \in A$, we take an upwards vertical ray starting in the point $P_a$ with $x$-coordinate equal to 
$y_a$ and with $y$-coordinate equal to $x_a$. For each $b \in B$, we take a horizontal ray to the right, starting in 
the point $Q_b$ with $x$-coordinate $x_b$ and $y$-coordinate $z_b$. Now $P_a$ intersects $Q_b$ if and only if 
$x_b \leq y_a$ and $x_a \leq z_b$, i.e., if and only if $ab \in E(H)$ as required.

Now suppose that $H$ has a two-directional model, i.e., upwards vertical rays $U_a, a \in A,$ and horizontal rays 
to the right $R_b, b \in B,$ such that $ab \in E(H)$ if and only if $U_a \cap R_b \neq \emptyset$. We will prove that
$H$ has a min ordering, whence it is a signed-interval digraph by Theorem \ref{main}. We will define the orders 
$<$ on $A$ and on $B$ as follows. Assume the starting point of the vertical ray $U_a$ has the $(x,y)$-coordinates 
$(u_a,v_a)$, and the starting point of the horizontal ray $R_b$ has the $(x,y)$-coordinates $(r_b,s_b)$, for $a \in A,$ 
and $b \in B$. It is easy to see that we may assume, without loss of generality, that all $u_a, a \in A,$ and $r_b, b \in B$ 
are distinct, and similarly for $v_a, a \in A$ and $s_b, b \in B$. We define $a < a'$ in $A$ if and only if $v_a < v_a'$, 
and define $b < b'$ in $B$ if and only if $r_b < r_{b'}$. We show that this is a min ordering of the bipartite digraph 
$H$. Otherwise, some $ab \in E(H), a'b' \in E(H), a < a', b' < b$ have $ab' \not\in E(H)$. There are two possibilities 
for $ab' \not\in E(H)$; either $u_a < r_{b'}$ or $u_a > r_{b'}, v_a > s_{b'}$. In the former case, $U_a \cap R_b = \emptyset$, 
in the latter case $U_{a'} \cap R_{b'} = \emptyset$, contradicting the assumptions. 
\end{proof}

\section{Algorithms and characterizations}

Interval graphs are known to have elegant characterization theorems \cite{fg,lekk}, cf. \cite{gol,spin} and efficient 
recognition algorithms \cite{booth,corneil,habib}. 
Thus one might hope to be able to obtain similar results for their generalizations and digraph analogues. This is true
for all the generalizations described in this paper, at least to some degree. In this section we summarize what is known.

The prototypical characterization of interval graphs is the theorem of Lekkerkerker and Boland \cite{lekk}. In our language,
it states that {\em a reflexive graph $H$ is an interval graph if and only if it contains no asteroidal triple and no induced
$C_4$ or $C_5$}. An {\em asteroidal triple} consists of three non-adjacent vertices such that any two are joined by a 
path not containing any neighbors of the third vertex. An equivalent characterization by the absence of a slightly less 
concise obstruction is given in \cite{adj}. {\em A reflexive graph $H$ is an interval graph if and only if it contains no
invertible pair.} An {\em invertible pair} is a pair of vertices $u, v$ such that there exist two walks of equal length, $P$ 
from $u$ to $v$, and $Q$ from $v$ to $u$, where the $i$-th vertex of $P$ is non-adjacent to the $(i+1)$-st vertex of $Q$
(for each $i$), and also two walks of equal length $R, S$ from $v$ to $u$ and $u$ to $v$ respectively, where the $i$-th
vertex of $R$ is non-adjacent to the $(i+1)$-st vertex of $S$ (for each $i$). It is not difficult to see that an asteroidal triple
is a special case of an invertible pair. A number of variants of the definition of an invertible pair have arisen \cite{adj,ross,jing,esa2012}, 
and they have proved useful to give characterization theorems for various classes. It is proved in \cite{adj} that 
{\em a reflexive digraph is an adjusted interval digraph if and only if it contains no directed invertible pair.}  A {\em 
directed version} of an invertible pair is defined in \cite{adj} in a manner similar to the above definition of an invertible 
pair. With yet another {\em labeled version} of an invertible pair, we have the following obstruction characterization
of co-TT graphs: {\em a graph is a co-TT graph if and only if it contains no labeled invertible pair}, which follows
from the characterization in~\cite{ross} in terms of an interval ordering from~\cite{circ}.
For bipartite graphs, an analogous {\em bipartite version} of an invertible pair yields the following result.
{\em A bipartite graph is a two-directional orthogonal ray graph if and only if it contains no bipartite invertible pair}, \cite{esa2012}.
In fact, in \cite{fhh} a stronger version is shown: there is a bipartite analogue of an asteroidal triple, called an {\em edge-asteroid},
and {\em a bipartite graph is a two-directional orthogonal ray graph if and only if it contains no induced $6$-cycle and no edge-asteroid}.
Finally, in \cite{arxiv}, there is an obstruction characterization for signed-interval digraphs, which is a little more technical than
just an invertible pair, \cite{arxiv}.  

There is a long history of efficient algorithms for the recognition of interval graphs, many of them linear time, starting from 
\cite{booth} and culminating in \cite{corneil}. A polynomial time algorithm for the recognition of adjusted interval digraphs 
is given in \cite{adj}. It is not known how to obtain a linear time, or even near-linear time algorithm. An $O(n^2)$ algorithm 
for the recognition of two-directional orthogonal ray graphs follows from Theorem \ref{bipo} and \cite{circ}. A more efficient 
algorithm in this case is also not known. On the other hand, an $O(n^2)$ algorithm for the recognition of co-TT graphs 
has been given in \cite{ross}. In~\cite{arxiv}, a polynomial-time algorithm
for the recognition of a signed-interval digraph is proposed.  (A new
version of~\cite{arxiv} will be posted on arXiv soon.)


\vspace{4mm}
\end{document}